\documentclass{article}
\usepackage[caption=false,font=footnotesize]{subfig}
\usepackage{hyperref}
\usepackage{amsmath,amssymb,amsthm}
\usepackage[a4paper,scale=0.75,includefoot,centering]{geometry}
\usepackage[dvipdfmx]{graphicx}
\usepackage{indentfirst}
\usepackage{cite}

\newtheorem{corollary}{\bf Corollary}[section]

\def\figref#1{\ref{#1}}
\def\eqref#1{(\ref{#1})}
\def\tabref#1{\ref{#1}}
\def\secref#1{\ref{#1}}

\def\Gamma{J}

\def\del{{\partial}}
\def\dom{\mathop{\rm dom}}
\def\E{\mathcal{E}}
\def\T{\overline{T}}
\def\V{\tilde{V}}
\def\VS{V_{\rm S}}
\def\VD{V_{\rm D}}
\def\ov#1{\overline{#1}}
\def\vct#1{{\mathchoice{\mbox{\boldmath$#1$}}{\mbox{\boldmath$#1$}}%
  {\mbox{\scriptsize\boldmath$#1$}}{\mbox{\scriptsize\boldmath$#1$}}}}

\title{Power packet transferability via symbol propagation matrix}

\author{
Shinya Nawata\footnote{Department of Electrical Engineering, Kyoto University, Katsura, Nishikyo, Kyoto 615-8510, Japan}, Atsuto Maki\footnote{School of Computer Science and Communication, Royal Institute of Technology (KTH), Teknikringen 14, Stockholm, 100 44 Sweden} and Takashi Hikihara$^{\ast}$}

\date{}

\begin{document}




\maketitle

\begin{center}
 {\bf Abstract}
\end{center}

Power packet is a unit of electric power transferred by a power pulse with an information tag.
In Shannon's information theory, messages are represented by symbol sequences in a digitized manner.
Referring to this formulation, we define symbols in power packetization as a minimum unit of power transferred by a tagged pulse.
Here, power is digitized and quantized.
In this paper, we consider packetized power in networks for a finite duration, giving symbols and their energies to the networks.
A network structure is defined using a graph whose nodes represent routers, sources, and destinations.
First, we introduce symbol propagation matrix (SPM) in which symbols are transferred at links during unit times.
Packetized power is described as a network flow in a spatio-temporal structure.
Then, we study the problem of selecting an SPM in terms of transferability, that is, the possibility to represent given energies at sources and destinations during the finite duration.
To select an SPM, we consider a network flow problem of packetized power.
The problem is formulated as an M-convex submodular flow problem which is known as generalization of the minimum cost flow problem and solvable.
Finally, through examples, we verify that this formulation provides reasonable packetized power.

\begin{center}
 {\bf Keywords}
\end{center}
\begin{center}
 power packet, router, network flow problem, electrical energy network
\end{center}

\section{Introduction}

Electric power has been considered as a continuous flow based on circuit theory, in which power flow is governed by Kirchhoff laws and Tellegen's theorem \cite{bib:desoer1969b}. 
The circuit theory can be generalized to represent various nonlinear complex systems in the system topology with energy dissipation and energy storage as network thermodynamics \cite{bib:oster1971a}.
Here, energy flow is handled in a continuous manner under the conservation of energy.
On the other hand, it is shown in Shannon's information theory \cite{bib:shannon1948a} that ``all technical communications are essentially digital; more precisely, that all technical communications are equivalent to the generation, transmission, and reception, of random binary digits'' \cite{bib:massey1984a}.
Communication networks have been developed in a digitized manner by utilizing packet switching, which breaks messages into smaller pieces named ``packets'', for dynamic assignment of network resources \cite{bib:kleinrock2010a}.  
If we handle electric power in a digitized manner, power distribution will be changed completely different from the conventional.
In this paper, we consider electrical energy networks in which power is digitized and quantized through power packetization \cite{bib:takuno2010i,bib:takahashi2012a,bib:takahashi2015a,bib:zhou2015b,bib:takahashi2016a,bib:mochiyama2016a,bib:nawata2014a_eng,bib:nawata2016a_upstream,bib:hikihara2016i}.

The concept of power packet was proposed in 1990s to manage complicated power flows in power systems caused by various power transactions after deregulation \cite{bib:toyoda_saitoh1998i}.
In the proposal, electric energy routers, which include energy storage devices, were installed into the electrical energy networks.
The router manipulates its own storage device according to the flow control data transferred with the power packet so as to compensate for the difference between the generation schedule and the demand schedule. 
Referring to this work, power packet transactions were proposed for an electric power distribution system \cite{bib:inoue2011a}.
He {\it et al.}  proposed an electric power architecture, rooted in lessons learned from the Internet and microgrids, to produce a grid network designed for distributed renewable energy, prevalent energy storage, and stable autonomous systems \cite{bib:he_katz2008i}.
``Energy packet networks'' were also proposed to provide energy on demand to Cloud Computing servers \cite{bib:gelenbe2012i}.
There is a proposal for controllable-delivery power grid in which electrical power is delivered through discrete power levels directly to customers \cite{bib:Rojas2013a2}. 
In the physical layer, a universal power router is designed and evaluated for residential applications \cite{bib:stalling2012i}.
On the other hand, in most of these proposals, electric energy and information are separately transferred or the physical design is not mentioned.
It has been difficult to realize the practical hardware to deal with electric power in the same way as information, because energy has been transferred with high-power and low-frequency electricity while information has been transferred with low-power and high-frequency electricity.
For ensuring consistency between the physical layer and the logical layer, the synchrony between energy and information is crucial for managing power.

Recently, wide-bandgap power semiconductors, such as Silicon Carbide (SiC) and Gallium Nitride (GaN), have shown material properties enabling power device operation at potentially higher temperatures, voltages, and switching speeds than the current Si technology \cite{bib:funaki2007a,bib:millian2014a}.
High-speed gate drive circuits have been developed to achieve high frequency switching over 1~MHz \cite{bib:takuno2009i,bib:nagaoka2015a}.
This work enables us to handle high-power and high-frequency electricity,
and moreover, to digitize power with power packetization \cite{bib:takuno2010i,bib:takahashi2015a,bib:zhou2015b}.
In the developed system, an information tag is directly and physically attached to each packet with its voltage waveform.
In this way, energy and information are integrated at an individual packet level.
A schematic of the power packet dispatching network is shown in Fig.~\figref{fig:network}.
The system consists of network lines and routers which stores and forwards power packets according to the tag's information.
In the network, the information tag identifies different kinds of power due to different sources, destinations, voltages, control commands, and so on.
When we send the power packets using time-division multiplexing, it becomes possible to distinguish the power at each line by using the information tag as an index. 
The router has 
multiple storage units
to identify the different kinds of power flow.
Here, power is given by the density of power packets between routers \cite{bib:nawata2016a_upstream,bib:takahashi2016a,bib:hikihara2016i}.

\begin{figure}[tbp]
 \centering
 \includegraphics[width=0.6\linewidth]{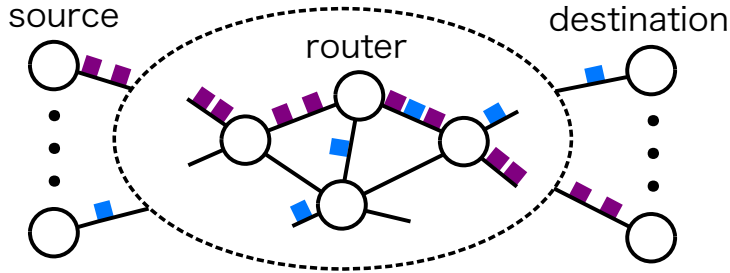}
 \caption{A schematic of power packet dispatching network.}
   \label{fig:network}
\end{figure}

Now, we consider a packet-centric framework of energy transfer.
In Shannon's information theory, messages are represented by symbol sequences in a digitized manner \cite{bib:shannon1948a}.
Referring to this formulation, we define symbols in power packetization as a minimum unit of power transferred by a power pulse with an information tag \cite{bib:nawata2014a_eng}. 
Because a symbol is a minimum unit of power, we ignore the variety of ways in which a symbol can be transferred during a unit time in the physical layer.
Then, energy of each symbol is uniquely determined as a real number\footnote{A load can be treated as resistive with PFC circuits. Thus, power is discussed in real numbers without loss of generality.}\footnote{In this setting, power pulses with the same energy can be represented by a single symbol. The properties which symbols do not specify are treated by indexing the symbols. The index becomes important in terms of redundancy of the system.}.
Here, power packetization is a simultaneous representation of messages and energy with symbol sequences \cite{bib:nawata2014a_eng}.
In information theory, the representation of messages is treated as a coding problem, in which the goal is to minimize the cost such as the length of codewords. 
In power packetization, however, it is important to represent the given energy during a finite duration as the total amount of energy of symbols.
Thus, energy representation is a problem unique to power packetization.
In \cite{bib:nawata2014a_eng}, this problem was considered with a set of symbol sequences which represent a given energy as the total amount.

In this paper, we consider packetized power in networks for a finite duration, giving the set of symbols and their energies to the networks.
Then, packetized power is spatially and temporally transferred as symbols in a digitized and quantized manner:
a symbol is transferred at a link during a unit time;
at each node, energy is represented with symbols sent to and received from neighboring nodes during a finite duration.

To mathematically represent the transmission of packetized power, we refer to the work about detecting bipedal motion using point trajectories in video sequences \cite{bib:maki2009i,bib:maki2013a}.
In this work, to obtain discriminative point trajectories from image sequences over a sufficiently long time period under both image noise and occlusion, probabilistic trajectories are designed by prioritizing the concept of temporal connectedness. 
They are extracted from directed acyclic graphs whose edges represent temporal point correspondences and are weighted with their matching probability in terms of appearance and location. 

Here, we introduce 
{\it symbol propagation matrix} (SPM), a new concept to represent packetized power, 
considering transfer of a symbol as a spatio-temporal correspondence.
In power packetization, unlike the probabilistic trajectories \cite{bib:maki2009i,bib:maki2013a}, each symbol has its energy and packetized power is represented as network flow.
The temporal connectedness is important in power packetization to transfer power in networks with low ``strain'', i.e. the spatial difference of power, which is equal to the temporal change of energy stored in each router.
Then, we consider the problem of selecting an SPM in terms of transferability, that is, the possibility to represent given energies at sources and destinations during the finite duration.
To select an SPM, we consider a network flow problem of packetized power, weighting supplied energy from sources and supplied energy to destinations  (V1), transferred energy at each link during each unit time (V2), and change of stored energy in each router (V3). 
The problem is formulated as an M-convex submodular flow problem which is known as generalization of the minimum cost flow problem and solvable \cite{bib:murota2003b}.
Finally, through examples, we verify that the formulation provides reasonable transmission of packetized power.

\section{Packetized power in networks}\label{sect:spm}

Here, we introduce symbol propagation matrix as a representation of packetized power transferred by symbols  in a digitized and quantized manner.
Via SPM, packetized power is represented as network flow in a spatio-temporal structure. 

\subsection{Symbol propagation matrices}

The set of symbols $\varSigma_{\rm T}$ and their energies $\mathcal{E}: \varSigma_{\rm T} \rightarrow \mathbb{R}_{>0}$ are given to the network. 
The symbols have a partition $\{\varSigma_{m}\}_{m=0}^{M-1}$,
whose cell represents a distinct power flow\footnote{In power packetization, energy is transferred with time division multiplexing (TDM) at links and  stored in the corresponding storage unit in routers. Thus, power flow can be distinguished by the information tags of power packets.}.
For each distinct power flow, energy is represented as a summation of the symbol's energy.
Here, symbols of the same cell can be exchanged under the conservation of energy.

The network structure is given as a directed graph $G=(\V, A)$, where $\V$ is a disjoint union of the set of routers $V$, sources $\VS$, and destinations $\VD$, and $A$ is the set of links.
Here, sources and destinations represent the external system of the network.
The incidence relation is a couple of functions $\partial^{+}: A \rightarrow \tilde{V}$ and $\partial^{-}: A \rightarrow \tilde{V}$.
Other representation of the incidence relation is introduced as a couple of functions 
$\delta^{+}: \tilde{V}\ni v\mapsto \{ a \mid  \partial^{+} a = v \} \in 2^{A}$ and $\delta^{-}: \tilde{V}\ni v\mapsto \{ a \mid  \partial^{-} a = v \} \in 2^{A}$ \cite{bib:iri1969b,bib:iri_fujishige_oyama1986,bib:murota2003b}.
As for the link, power is kept between nodes.
The directions of links are assigned with power directions.

Next, we set that the network is synchronized and 
symbols are transferred during the same unit times $\T=\{\ov{t_0}, \ov{t_1}, \cdots, \ov{t_{N-1}}\}$, 
where $N$ is a positive integer, $\ov{t_n}:= [t_n, t_{n+1})$, and $t_0 < t_1 < \cdots < t_N$.
Here, energy is transferred by $N$ unit times.
Although various power pulses can transfer the energy of the same symbol during a unit time, we ignore the variety and focus on the integrated value of power during the unit time.

Now, we focus on a single cell $\varSigma_{m}$.
At each link $a \in A$ during each unit time $\ov{t}\in \T$, 
there are three cases of transfer of symbols $\varSigma_{m}$: 
Case 1: a single symbol $\sigma\in \varSigma_{m}$ is transferred from node $\del^{+}a$ to node $\del^{-}a$;
Case 2: a single symbol $\sigma\in \varSigma_{m}$ is transferred from node $\del^{-}a$ to node $\del^{+}a$;
Case 3: no symbol is transferred\footnote{Here, because we do not care the various ways of power transfer during each unit time, more than one symbol is not transferred.
For example, if twice the amount of energy of a symbol $\sigma_1 \in \varSigma_m$ is transferred, then we consider that a symbol $\sigma_2 \in \varSigma_m$ with a energy $\E(\sigma_2) = 2\E(\sigma_1)$ is transferred.}. 
Therefore, packetized power is given by a map, which we name 
{\it symbol propagation matrix}:
\begin{equation}
  SPM_{m}: \overline{T}\times A \rightarrow \varSigma_{m} \times \{{\rm f}, {\rm b}\} \cup \{\sigma_{\emptyset}\},
\end{equation}
where $\sigma_{\emptyset}$ is an element which does not belong to $\varSigma_{\rm T}$.
$SPM_{m}(\ov{t}, a) = (\sigma, {\rm f})$ and $SPM_{m}(\ov{t}, a) = (\sigma, {\rm b})$ denote that $\sigma$ is transferred from node $\del^{+}a$ to node $\del^{-}a$ and from node $\del^{-}a$ to node $\del^{+}a$, respectively. 
$SPM_{m}(\ov{t}, a) = \sigma_{\emptyset}$ denotes that no symbol is transferred.

\begin{figure*}[t!]
 \centering
 \subfloat[]{\includegraphics[width=50mm]{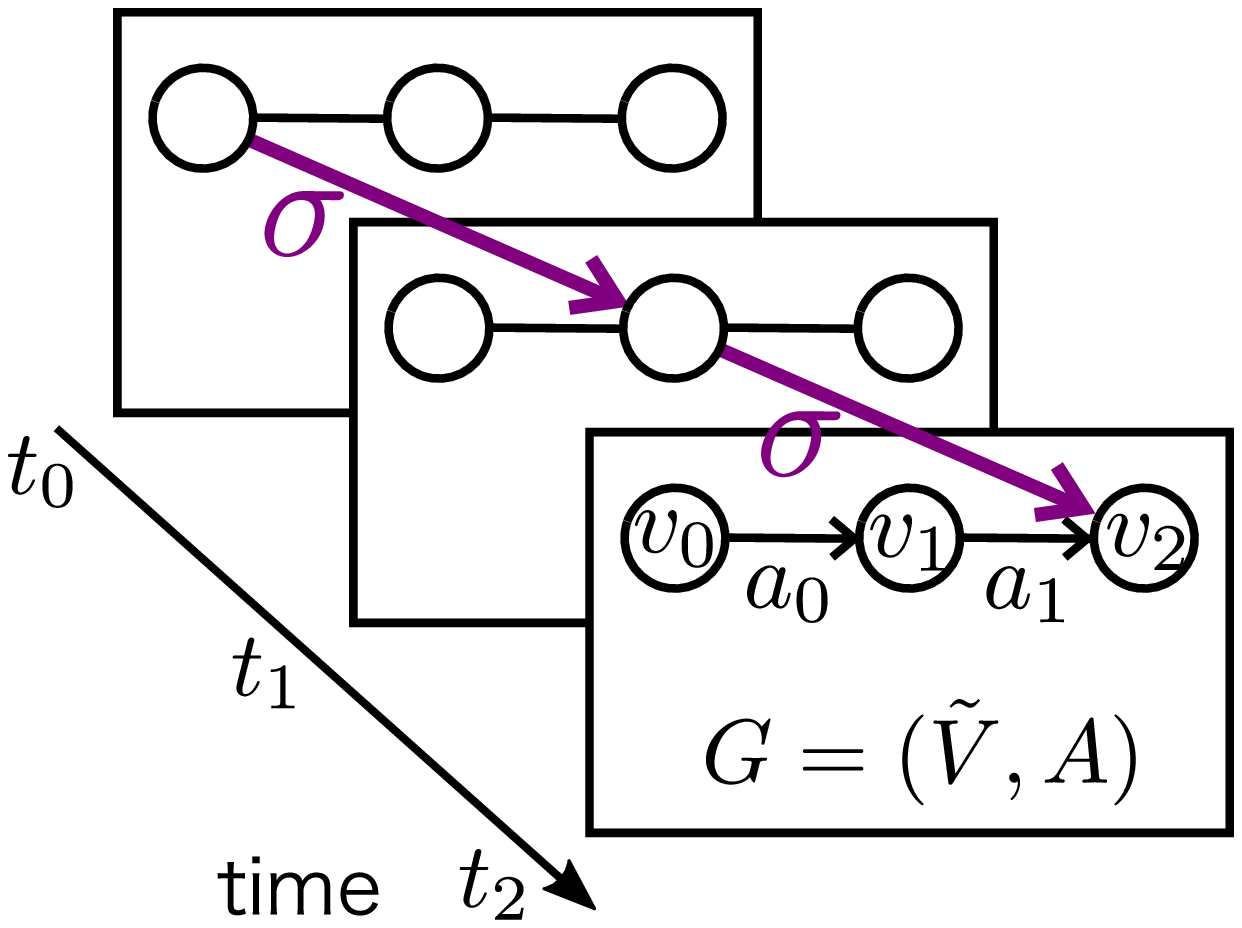}\label{fig:symbol_propagation}}
\hfil
 \subfloat[]{\includegraphics[width=23mm]{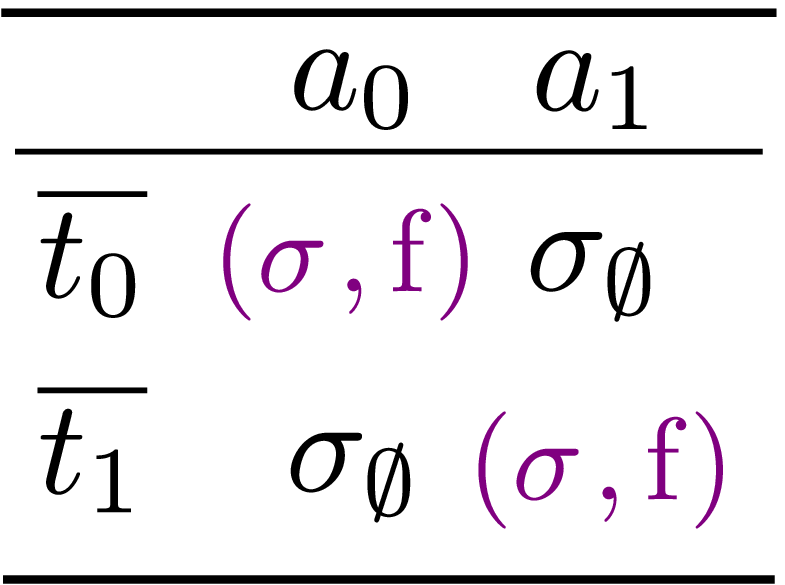}\label{fig:spm}}
\hfil
 \subfloat[]{\includegraphics[width=45mm]{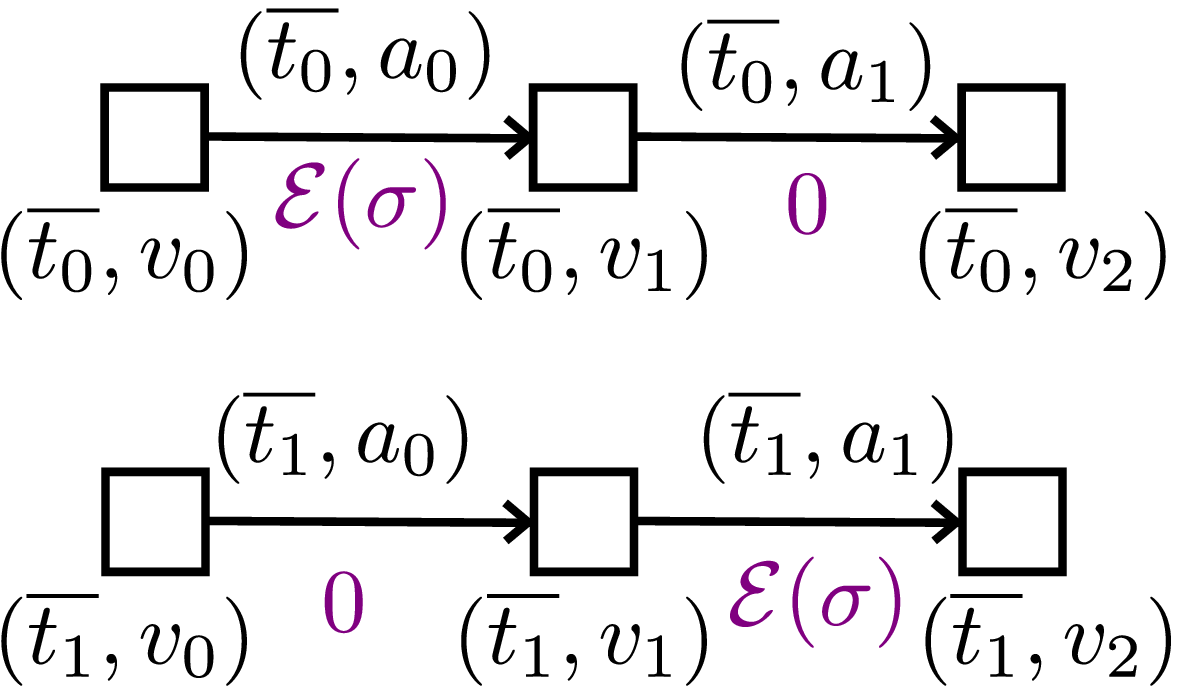}\label{fig:network_flow}}
 
 \caption{Schematic of symbol propagation matrix and packetized power. We set a directed graph $G=(\V, A)$, where $\V=\{v_0, v_1, v_2\}$ and $A=\{a_0, a_1\}$, and unit times $\T=\{\ov{t_0}, \ov{t_1}\}$. (a) A schematic diagram of symbol propagation. A symbol $\sigma\in \varSigma_m$ is transferred at link $a_0$ during unit time $\ov{t_0}$ and at link $a_1$ during unit time $\ov{t_1}$. (b) Symbol propagation matrix shown as a table. (c) Packetized power as network flow on $\hat{G}=(\T\times \V, \T \times A)$.} 
 \label{fig:spm_flow}
\end{figure*}

\subsection{Packetized power as network flow}\label{sec:network_flow}

Here, we introduce packetized power as network flow \cite{bib:iri1969b,bib:iri_fujishige_oyama1986,bib:murota2003b}.
First, we define a graph with spatio-temporal structure induced by the network structure $G$ and unit times $\T$ as
\begin{equation}
 \hat{G}=(\T\times\V, \T\times A),
\end{equation}
whose incidence relation is defined by
\begin{equation}
 \hat{\del}^{+}:\T\times A\ni(\ov{t}, a)\mapsto (\ov{t}, \del^{+}a) \in \T \times \V
\end{equation}
and
\begin{equation}
 \hat{\del}^{-}:\T\times A\ni(\ov{t}, a)\mapsto (\ov{t}, \del^{-}a) \in \T \times \V.
\end{equation}
Then, packetized power is introduced as a network flow on $\hat{G}$, i.e. 
\begin{equation}
 u: \T \times A \rightarrow \mathbb{R},
\end{equation}
and its boundary $\del u: \T \times \V \rightarrow \mathbb{R}$ is defined by 
\begin{equation}\label{eq:conservation}
 \del u(\ov{t}, v) = \sum_{a\in \delta^{+}v} u(\ov{t}, a) - \sum_{a\in \delta^{-}v} u(\ov{t}, a)
 \quad (\ov{t}\in \overline{T}, v\in \tilde{V}).
\end{equation}

Each link of $\hat{G}$, i.e. $(\ov{t}, a)\in \T\times A$, represents a spatio-temporal correspondence at which each symbol is transferred.
Each cell $\varSigma_m$ represents packetized power ${u}_{m}: \T \times A \rightarrow \mathbb{R}$ and $u_m$ is given by $SPM_{m}$ as
\begin{equation}
 {u}_{m}(\ov{t}, a) = \begin{cases}
		   \E(\sigma) & (SPM_{m}(\ov{t}, a) = (\sigma, {\rm f})),\\
		   -\E(\sigma) & (SPM_{m}(\ov{t}, a) = (\sigma, {\rm b})),\\
		    0  & (SPM_{m}(\ov{t}, a) = \sigma_{\emptyset}).
		  \end{cases}
\end{equation}
Here, we assume the conservation of energy; more precisely, we assume that all energy exchanges in each node $v\in \V$ are represented with symbols transferred through the adjacent links\footnote{In practical systems, energy can be dissipated at links and in nodes, and symbols may not be able to keep constant energy between nodes. These can become noise of power packetization, which is one of future work of this paper.} $a\in \delta^{+}v \cup \delta^{-}v$.  
Then, $-\del {u}_{m}(\ov{t}, v)$ is equal to the increment of stored energy corresponding to $\varSigma_m$ in the node $v\in \V$ during the unit time $\ov{t}\in \T$.
Here, packetized power is represented as network flow, in which energy is transferred with symbols in the digitized and quantized form.

\subsection{Example of symbol propagation matrix}

Here, we illustrate the definitions mentioned above.
Figure~\figref{fig:spm_flow} shows a schematic of a symbol propagation matrix and packetized power.
In this example, we set a directed graph $G=(\V, A)$, where $\V=\{v_0, v_1, v_2\}$ and $A=\{a_0, a_1\}$, and unit times $\T=\{\ov{t_0}, \ov{t_1}\}$.
As schematically shown in Fig.~\figref{fig:spm_flow}\subref{fig:symbol_propagation}, 
we consider a transmission of symbols $\varSigma_m$, in which a symbol $\sigma\in \varSigma_m$ is transferred at link $a_0$ during unit time $\ov{t_0}$ and at link $a_1$ during unit time $\ov{t_1}$.
Then, we have the SPM shown in Fig.~\figref{fig:spm_flow}\subref{fig:spm}.

In this example, the graph $\hat{G}=(\T\times \V, \T\times A)$ with spatio-temporal structure is defined as shown in Fig.~\figref{fig:spm_flow}\subref{fig:network_flow}.
Then, packetized power $u_m$, which is represented by the cell $\varSigma_m$, is introduced as a network flow on this graph.
Here, we have $u_{m}(\ov{t_0}, a_0)=\E(\sigma)$, $u_{m}(\ov{t_0}, a_1)=0$, $u_{m}(\ov{t_1}, a_0)=0$, and $u_{m}(\ov{t_1}, a_1)=\E(\sigma)$.
The boundary $\del u_m$ is equal to the change of stored energy corresponding to $\varSigma_m$.
For example, we have $-\del u_m(\ov{t_0}, v_1) = \E(\sigma)$ as the ``strain'', i.e. as the spatial difference of the packetized power $u_m$, and 
this value is equal to the increment of the stored energy in $v_1$ during the unit time $\ov{t_0}=[t_0, t_1)$.

\section{Power packet transferability}\label{sect:msfp3}

To meet supply and demand, i.e. to represent given energies at sources and destinations with symbols,
it is necessary to select a symbol at each link $a\in A$ during each unit time $\ov{t}\in \T$. 
Here, different choices lead to different transferability, that is the possibility to represent the given energies by transmission of symbols.
In this section, we develop a framework to select an SPM in terms of transferability as a network flow problem of packetized power.
To make the problem solvable, we focus on a single power flow represented by cell $\varSigma \in \{\varSigma_{m}\}_{m=0}^{M-1}$. 
Besides, we set that $\E(\varSigma)$ to be successive integers $\{1, 2, \cdots\}$, that is, we consider integer flow $u:\T\times A\rightarrow \mathbb{Z}$.
Because power packet is a unit of power, it is natural to consider integer flow.

\subsection{Features of packetized power contributing to transferability}

In networks, packetized power appears as: 
\begin{itemize}
 \item[V1:] supplied energy from sources and supplied energy to destinations,
 \item[V2:] transferred energy at each link during each unit time, and
 \item[V3:] change of stored energy in each router.
\end{itemize}
In terms of transferability, V1 needs to satisfy the given energies at sources and destinations,
while V2 and V3 need to be small. 
As for V2, because power is given by the density of power packets at each links, 
transferred energy at links should be small to utilize limited density of packets during each unit time.
In addition, by minimizing the summation of transferred energy, we can obtain the network flow in which energy is supplied to each destination from sources placed near the destination.
As for V3, change of stored energy in routers should be suppressed to keep symbol's energy controllable with density modulation of power packets between routers.

Thus, we select the features of packetized power $u:\T\times A \rightarrow \mathbb{Z}$ contributing to transferability as V1, V2, and V3.
V2 is a value of the network flow $u(\ov{t}, a)$ $((\ov{t}, a)\in \T\times A)$, 
while V1 and V3 are calculated from the values of the boundary $\del u(\ov{t}, v)$ $((\ov{t}, v)\in \T\times \V)$.
V1 and V3 include time integral, such as total supplied energy from a source $s\in V_{\rm S}$, i.e. $\sum_{n=0}^{N-1}\del u (\ov{t_n}, s)$.

We introduce a cost function of the network flow problem as the summation of costs on these values.
Because our features include not only values of network flow but also values of boundary and their time integrals,
it is impossible to formulate the problem as the conventional minimum cost flow problem on the spatio-temporal graph $\hat{G}(\T\times \V, \T\times A)$.
Thus, in the next subsection, we provide the formulation using an M-convex submodular flow problem \cite{bib:murota2003b} which is the generalization of the minimum cost flow problem.

\subsection{Formulation as M-convex submodular flow problem}

Now, we formulate the network flow problem to provide transferable packetized power. 
In the spatio-temporal structure, the M-convex submodular flow problem is to find a packetized power $u:\T\times A \rightarrow \mathbb{Z}$ which minimizes the total cost $\Gamma(u)$;
more precisely, the problem is described by the graph $G(\T\times \V, \T\times A)$, univariate discrete convex functions\footnote{A function $\phi: \mathbb{Z}\rightarrow \mathbb{R}\cup \{+\infty\}$ is called a univariate discrete convex function if we have
\begin{equation}\label{eq:univariate}
 \phi(x-1)+\phi(x+1)\geq 2\phi(x) \mbox{ for all } x\in \mathbb{Z}
\end{equation}
and $\dom \phi \neq \emptyset$ \cite{bib:murota2003b}.
Note that, if a function $\phi:\mathbb{R}\rightarrow \mathbb{R}\cup \{+\infty\}$ is convex, $\phi$ satisfies \eqref{eq:univariate}.} $f_{(\ov{t}, a)}: \mathbb{Z}\rightarrow \mathbb{R}\cup \{+\infty\}$ $((\ov{t}, a)\in  \T \times A)$, and an M-convex function $f: \mathbb{Z}^{\T \times \V} \rightarrow \mathbb{R}\cup \{+\infty\}$ as in 
\cite{bib:murota2003b}
\begin{eqnarray}
  \!\!\mbox{Minimize} &&  \Gamma(u)=\!\!\sum_{(\ov{t}, a)\in \T\times A}\!\! f_{(\ov{t}, a)}(u(\ov{t}, a)) +  f(\del u) \label{eq:MSFP3_obj_func}\\
  \!\!\mbox{subject to} && u(\ov{t}, a)\in \dom f_{(\ov{t}, a)} \!\!\quad ((\ov{t}, a)\in \T\times A), \label{eq:MSFP3_dom_fta}\\
  && \del u \in \dom f, \label{eq:MSFP3_dom_f} \\
  && u(\ov{t}, a) \in \mathbb{Z} \quad ((\ov{t}, a)\in \T\times A)\label{eq:MSFP3_integer}.
\end{eqnarray}
Here, \eqref{eq:MSFP3_dom_fta} denotes capacity constraints of links, because $f_{(\ov{t}, a)}$ is convex, and hence $\dom f_{(\ov{t}, a)}$ is an interval.
Therefore, we can set an upper capacity $\hat{c}:\T \times A \rightarrow \mathbb{Z}\cup\{+\infty\}$ and a lower capacity $\check{c}:\T \times A \rightarrow \mathbb{Z}\cup \{-\infty\}$, where it is assumed that $\hat{c}(\ov{t}, a)\geq \check{c}(\ov{t}, a)$ for each $(\ov{t}, a)\in \T\times A$, by defining cost functions $f_{(\ov{t}, a)}$ whose effective domain is equal to the interval $[\check{c}(\ov{t}, a), \hat{c}(\ov{t}, a)]$  $((\ov{t}, a)\in \T \times A)$.  

Then, we introduce the cost function of the boundary as the summation of the costs on V1 and V3.
To this end, we prove in Sect.~\ref{sec:proof_of_M-convexity} that, for a laminar family\footnote{By a laminar family, we mean a nonempty family $\mathcal{T}$ such that \cite{bib:murota2003b} 
\begin{equation}
 X, Y \in  \mathcal{T} \Rightarrow X\cap Y = \emptyset \mbox{ or } X\subset Y \mbox{ or } X \supset Y.
\end{equation} } $\mathcal{T}$ of subsets of $\T \times \V$ and univariate discrete convex functions $f_X: \mathbb{Z}\rightarrow \mathbb{R}\cup \{+\infty\}$ indexed by $X\in \mathcal{T}$,
the function defined by
\begin{equation}\label{eq:node_cost}
 f(\Delta u) =
  \begin{cases}
   \sum_{X\in \mathcal{T}}f_{X}(\Delta u(X)) & ( \Delta u(\T \times \V) = 0 ),\\
   +\infty & (\mbox{otherwise})
  \end{cases}
\end{equation}
is an M-convex function.
Here, we use the notation $\Delta u (X) = \sum_{(\ov{t}, v)\in X}\Delta u(\ov{t}, v) $ 
for $\Delta u \in \mathbb{Z}^{\T\times \V}$ and $X\subset \T\times \V$.
In \eqref{eq:node_cost}, we can set costs on V1 and V3 by setting a laminar family $\mathcal{T}$ and univariate discrete convex functions $\{f_X\}_{X\in \mathcal{T}}$.
For example, we can treat total supplied energy $\sum_{n=0}^{N-1}\del u(\ov{t_n}, s)$ at source $s\in V_{\rm S}$ by including $\T\times \{s\}$ in $\mathcal{T}$.
In the following, we set that the laminar family $\mathcal{T}$ is a disjoint union of a laminar family $\mathcal{T}_{\rm S, D}$ of subsets of $\T \times (V_{\rm S} \cup V_{\rm D})$  and laminar families $\mathcal{T}_{v}$ of subsets of $\T \times \{v\}$ $(v\in V)$.

To sum up, we introduce the cost function $\Gamma$ of packetized power $u$ as
\begin{eqnarray}\label{eq:MSFP3_obj_func_sep}
 \Gamma(u)&{}={}&\sum_{X\in \mathcal{T}_{\rm S, D}}f_{X}(\del u(X)) \nonumber\\
 && {}+{} \sum_{(\ov{t}, a)\in \T\times A}f_{(\ov{t}, a)}(u(\ov{t}, a)) \nonumber\\
 && {}+{} \sum_{v\in V}\sum_{X\in \mathcal{T}_{v}}f_{X}(\del u(X)).
\end{eqnarray}
The first, the second, and the third term in the right-hand side of \eqref{eq:MSFP3_obj_func_sep} correspond to the costs of V1, V2, and V3, respectively.

\subsection{Proof of the M-convexity of the function \texorpdfstring{$f$}{f} in \eqref{eq:node_cost}}\label{sec:proof_of_M-convexity}

In general, a laminar convex function has M$^{\natural}$-convexity\footnote{M$^{\natural}$-convex functions are variants of, and essentially equivalent to, M-convex functions. ''M$^{\natural}$-convex'' should be read ``M-natural-convex'' \cite{bib:murota2003b}.} \cite{bib:murota2003b}. 
The following corollary shows that, if a laminar convex function is restricted to the hyper-plane, the function has M-convexity.
From this corollary, we can confirm that the function $f$ in \eqref{eq:node_cost} satisfies M-convexity.
This is proved referring to Note 9.31. in \cite{bib:murota2003b} with a slight modification.

\begin{corollary}\label{note:M-covexity}
  For a finite set $V$, a laminar family $\mathcal{T}$ of subsets of $V$, univariate discrete convex functions $f_X$ $(X\in \mathcal{T})$, and an integer $r\in \mathbb{Z}$, a function $f:\mathbb{Z}^{V}\rightarrow \mathbb{R} \cup \{+\infty\}$ is defined by 
\begin{equation}
 f(x) = 
  \begin{cases}
    \sum_{X\in \mathcal{T}}f_{X}(x(X)) &  ( x(V)=r),\\
   +\infty                                          & (\mbox{otherwise}).
  \end{cases}
\end{equation}
Then, $f$ has M-convexity.

\end{corollary}

\begin{proof}[Proof of Corollary~\ref{note:M-covexity}]
Without loss of generality, we assume that $\emptyset \in \mathcal{T}$, $V\in \mathcal{T}$, and every singleton set belong to $\mathcal{T}$.
We represent $\mathcal{T}$ by a directed tree $G=(U, A; S, T)$ with root $u_0$,
where $U=\{u_{X} \mid X\in \mathcal{T}\}\cup\{u_0\}$, $A=\{a_{X}\mid X\in \mathcal{T}\}$, $S=\{u_0\}$, $T=\{u_{\{v\}}\mid v\in V\}$, 
and $\del^{-}a_{X} = u_{X}$ and $\del^{+}a_{X} = u_{\hat{X}}$ for $X\in \mathcal{T}$, where $\hat{X}$ denotes the smallest member of $\mathcal{T}$ that properly contains $X$ (and $\hat{V}=0$ by convention).
We associate the given function $f_{X}$ with arc $a_{X}$ for $X\in \mathcal{T}$.
We define an M-convex function $f':\mathbb{Z}^{S} \rightarrow \mathbb{R} \cup \{+\infty\}$ by
\begin{equation}
 f'(y)=\begin{cases}
	0 & (y=r), \\
	+\infty & (\mbox{otherwise}).
       \end{cases}
\end{equation} 
Then, $f$ is the result of flow type transformation:
\begin{eqnarray}
f(x)=\mathop{\rm inf}_{\xi, y}\!\!\!\!\!\!&&\!\!\!\!\Biggl\{ f'(y) + \sum_{X\in \mathcal{T}}f_{X}(\xi(a_{X})) \Biggm| \del \xi = (y, -x, \vct{0}), \Biggr.\nonumber\\
&&\Biggl.\xi\in\mathbb{Z}^{A}, (y, -x, \vct{0})\in \mathbb{Z}^{S}\!\times\!\mathbb{Z}^{T}\!\times\!\mathbb{Z}^{U\setminus(S\cup T)}\Biggr\}\nonumber \\
&& \qquad\qquad\qquad\qquad\qquad\qquad\quad\!\! (x\in \mathbb{Z}^{T}).
\end{eqnarray}
Therefore, from Theorem 9.27. in \cite{bib:murota2003b}, $f$ has M-convexity.  
\end{proof}

\section{Examples}\label{sect:example}

In this section, we verify that the formulation in Sect.~\secref{sect:msfp3} provides reasonable packetized power.
Here, the power packet transferability is discussed through
the following M-convex submodular flow problem:
\begin{eqnarray}
  \mbox{Minimize}\!\!\!\! &&  \Gamma(u)=
   \!\!\sum_{V' \in \mathcal{T}_{V_{\rm S} \cup V_{\rm D}}}\!\! f_{\T\times V'}(\del u(\T \times V'))\nonumber \\ 
 && \qquad\qquad + \sum_{(\ov{t}, a)\in \T\times A} \beta(a) \left| u(\ov{t}, a) \right|
   \label{eq:example_obj_func}\\
  \mbox{subject to}\!\!\!\! && \check{c}(a) \leq u(\ov{t}, a) \leq \hat{c}(a) \quad\!\! ((\ov{t}, a)\in \T\times A), \label{eq:example_dom_fta}\\
  && \del u (\ov{t}, v) = 0 \quad ((\ov{t}, v)\in \T \times V), \label{eq:example_dom_f} \\
  && u(\ov{t}, a) \in \mathbb{Z} \quad ((\ov{t}, a)\in \T\times A)\label{eq:example_integer},
\end{eqnarray}
where $\mathcal{T}_{V_{\rm S} \cup V_{\rm D}}$ is a laminar family of $V_{\rm S}\cup V_{\rm D}$. 
Note that $\{\T\times V' \mid V' \in \mathcal{T}_{V_{\rm S} \cup V_{\rm D}}\}$ is a laminar family of $\T \times (V_{\rm S} \cup V_{\rm D})$.
This problem is a special case in which we set the followings:
\begin{itemize}
 \item Total supplied energy is the only concern at sources and destinations. More precisely, $\mathcal{T}_{\rm S, D}$ is set to be $\{ \T \times V' \mid V' \in \mathcal{T}_{V_{\rm S} \cup V_{\rm D}}\}$.
 \item Each cost function $f_{(\ov{t}, a)}$ is defined as an absolute value function with effective domain  $[\check{c}(a), \hat{c}(a)]$ $((\ov{t}, a)\in \T\times A)$. The coefficients $\beta(a)$ $( a\in A)$ are non-negative real numbers.
       Here, $f_{(\ov{t}, a)}$ does not depend on unit times $\T$.
 \item Boundary of flow is set to be zero at routers $V$ in \eqref{eq:example_dom_f} in order to transfer power without strain.
       Note that, because of this constraint, we have $\del u(\{\ov{t}\}\times V_{\rm S}) = -\del u(\{\ov{t}\}\times V_{\rm D})$ $(\ov{t}\in \T)$ for a feasible flow $u$.

\end{itemize}  

\begin{figure}[tbp]
  \centering
   \includegraphics[width=50mm]{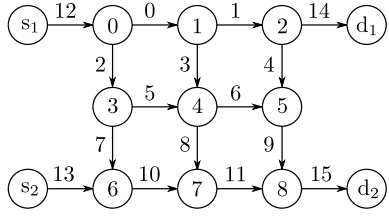}
 \caption{Network structure given as a mesh graph. For links between routers $a\in \{ 0, 1, \cdots, 11\}$, we set $\beta(a)=1$, $\check{c}(a) = -1$, and $\hat{c}(a)=1$, which imply $u(\ov{t}, a)\in \{-1, 0, 1\}$ $(\ov{t}\in \T)$.
For the other links $a\in \{ 12, 13, 14, 15\}$, we set $\beta(a)=0$, $\check{c}(a) = -\infty$ and $\hat{c}(a) = +\infty$, which imply $u(\ov{t}, a)\in \mathbb{Z}$ $(\ov{t}\in \T)$.}
 \label{fig:mesh_graph_3_3}
\end{figure}

In the following, we investigate the aforementioned problem with a mesh graph as an example.
The network structure is shown in Fig.~\figref{fig:mesh_graph_3_3}, where $V=\{0, \cdots, 8\}$, $V_{\rm S}=\{{\rm s}_1, {\rm s}_2\}$, $V_{\rm D}=\{{\rm d}_1, {\rm d}_2\}$, and $A=\{0, 1, \cdots, 15\}$.
For links between routers $a\in \{ 0, 1, \cdots, 11\}$, we set $\beta(a)=1$, $\check{c}(a) = -1$, and $\hat{c}(a)=1$.
This capacity constraint implies $u(\ov{t}, a)\in \{-1, 0, 1\}$ $(\ov{t}\in \T, a\in \{ 0, 1, \cdots, 11\})$.
For the other links $a\in \{ 12, 13, 14, 15\}$, we set $\beta(a)=0$, $\check{c}(a) = -\infty$, and $\hat{c}(a) = +\infty$, and hence we have $u(\ov{t}, a)\in \mathbb{Z}$ $(\ov{t}\in \T)$.

We consider three examples of the problem, in which energy is given at the sources and the destinations in different ways:
in the first example (E1), energy is given at each source and each destination by $U_{1}: V_{\rm S}\cup V_{\rm D} \rightarrow \mathbb{Z}$;
in the second example (E2), energy is given at each destination by $U_{2}: V_{\rm D}\rightarrow \mathbb{Z}_{\geq 0}$;
in the third example (E3), the total supplied energy is given by $U\in \mathbb{Z}_{\geq 0}$.
In E1, $U_{1}$ is set to satisfy $\sum_{s\in V_{\rm S}} U_{1}(s) = -\sum_{d\in V_{\rm D}} U_{1}(d)$ because power is transferred without strain in this example. 
In E1, E2, and E3, the objective functions are defined respectively as
\begin{eqnarray}
  \Gamma_{1}(u)&{}={}& \sum_{v\in V_{\rm S}\cup V_{\rm D}} 1000 \left| \del u(\T \times \{v\}) - U_1(v) \right|\nonumber\\ 
   && {}+{} \sum_{(\ov{t}, a)\in \T\times A}  \beta(a)\left| u(\ov{t}, a) \right|, \label{eq:case1_obj_func}
\end{eqnarray}
\begin{eqnarray}
  \Gamma_{2}(u)&{}={}&
  \sum_{v\in V_{\rm D}} 1000 \left| \del u(\T \times \{v\}) + U_2(v) \right| \nonumber\\
   && {}+{} \sum_{(\ov{t}, a)\in \T\times A}  \beta(a)\left| u(\ov{t}, a) \right|,\label{eq:case2_obj_func}
\end{eqnarray}
and
\begin{eqnarray}
  \Gamma_{3}(u)&{}={}&
  1000 \left| \del u(\T \times V_{\rm D}) + U \right| \nonumber\\ 
  &&{}+{} \sum_{(\ov{t}, a)\in \T\times A}  \beta(a)\left| u(\ov{t}, a) \right|.\label{eq:case3_obj_func}
\end{eqnarray}
In these functions, the coefficients of the first terms are set to be $1000$.
This value is large enough to give priority to the representation of the given energy over minimization of energy transferred through links.

\begin{figure*}[t]
 \centering
 \subfloat[An optimal flow $\alpha_0$]{\includegraphics[width=0.25\linewidth]{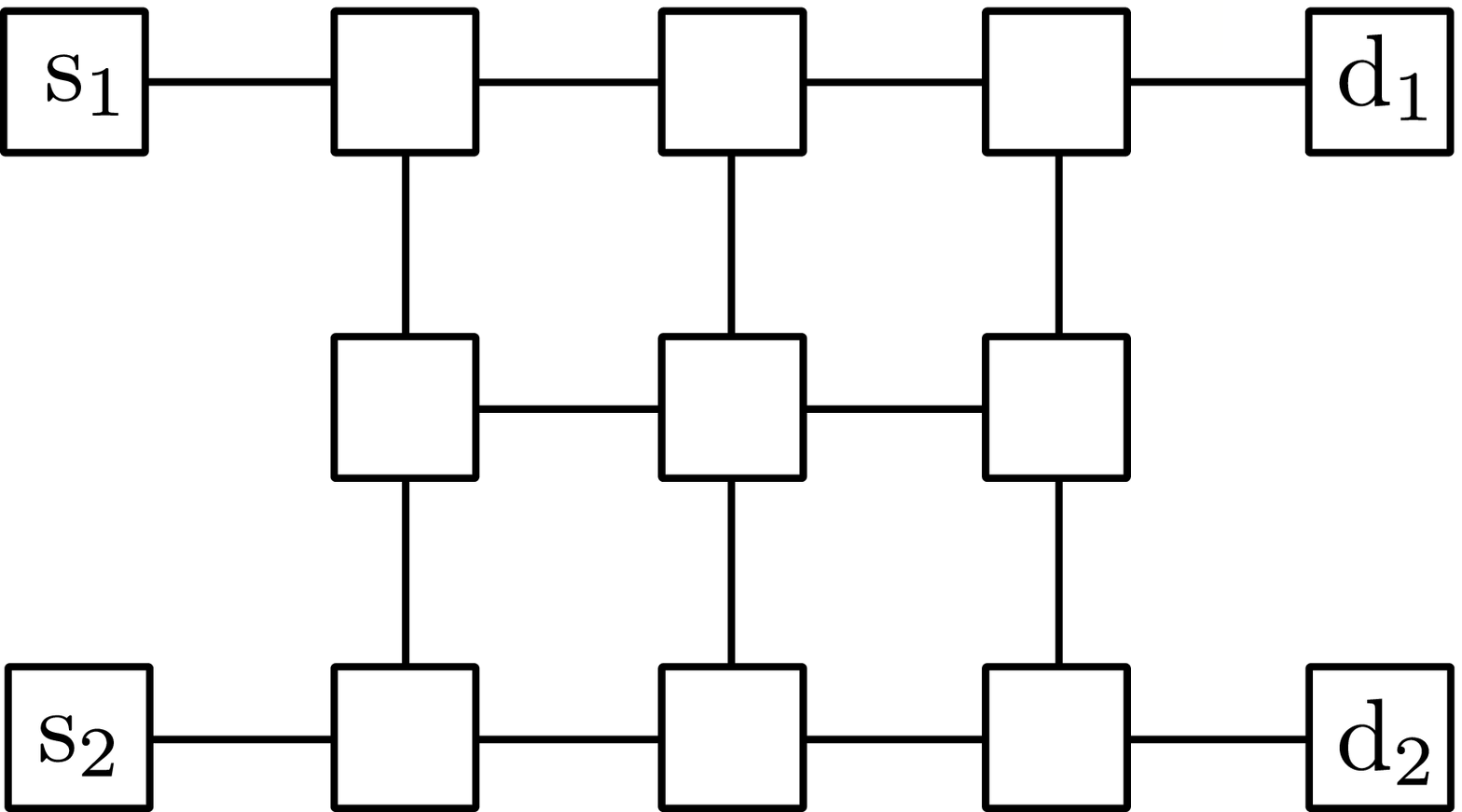}\label{fig:setting1-0}}
 \hfil
 \subfloat[An optimal flow $\alpha_1$]{\includegraphics[width=0.25\linewidth]{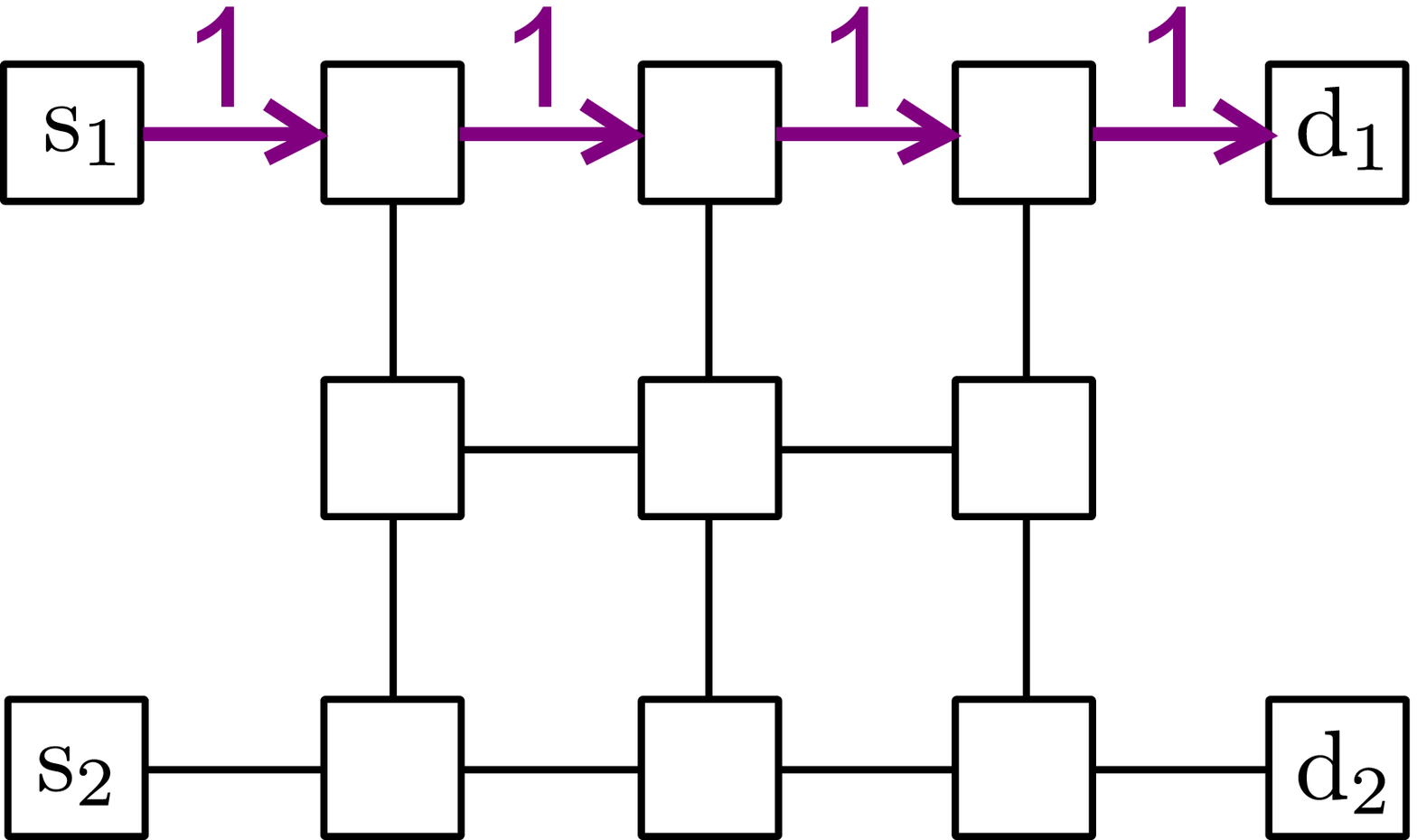}\label{fig:setting1-1}} 
 \hfil
  \subfloat[An optimal flow $\alpha_2$]{\includegraphics[width=0.25\linewidth]{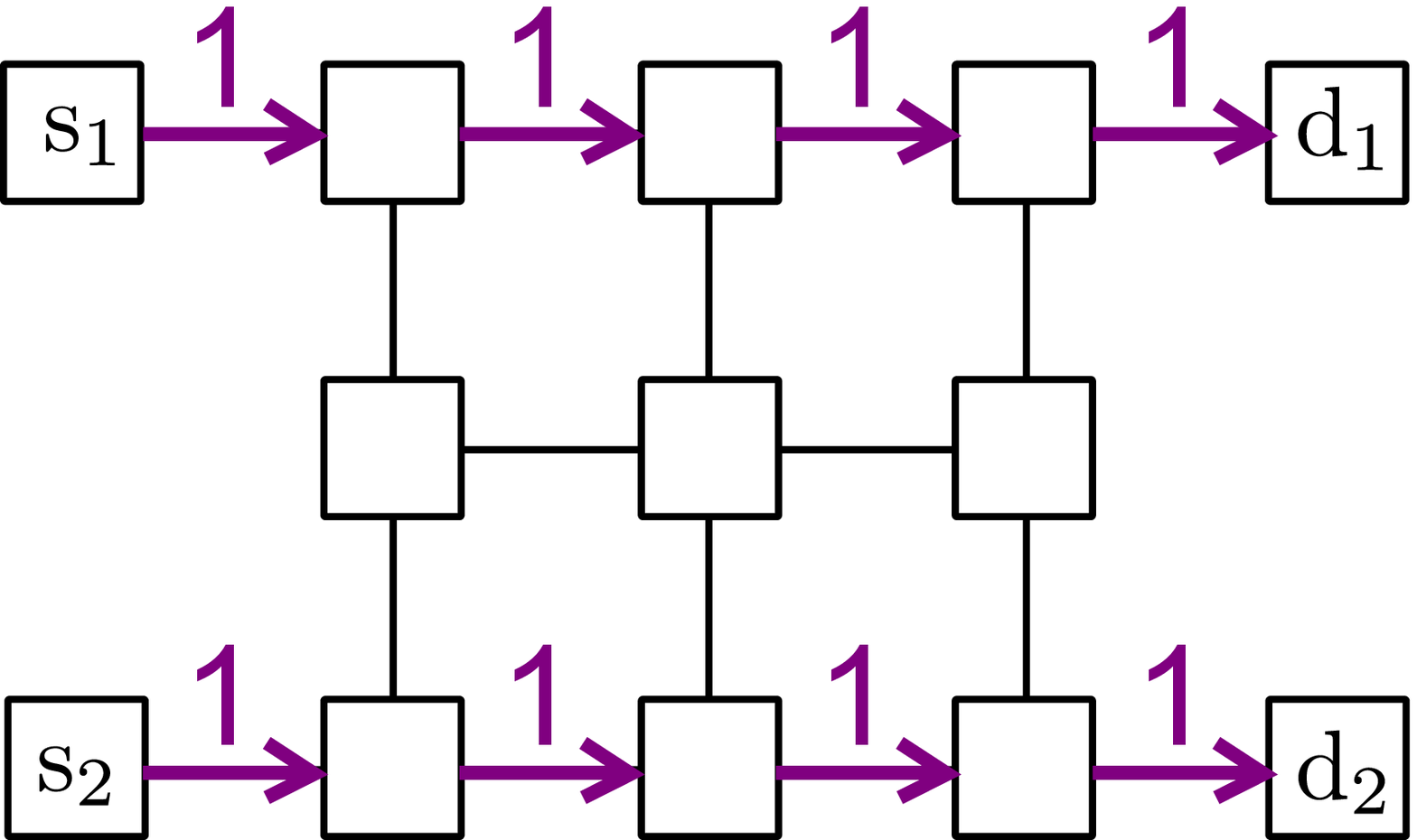}\label{fig:setting1-2}}
 
 \subfloat[An optimal flow $\alpha_3$]{\includegraphics[width=0.25\linewidth]{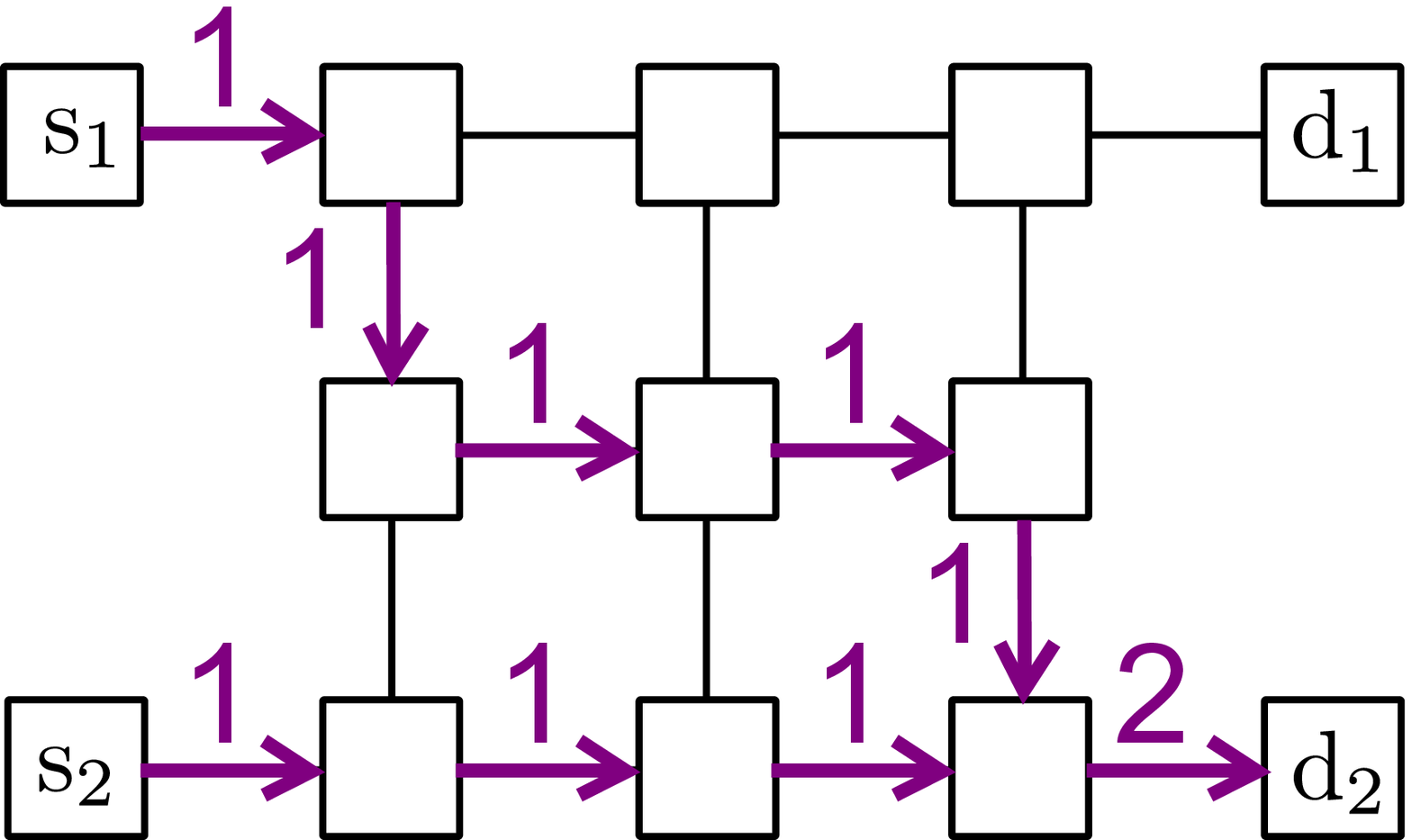}\label{fig:setting1-3}} 
 \hfil
 \subfloat[An optimal flow $\alpha_4$]{\includegraphics[width=0.25\linewidth]{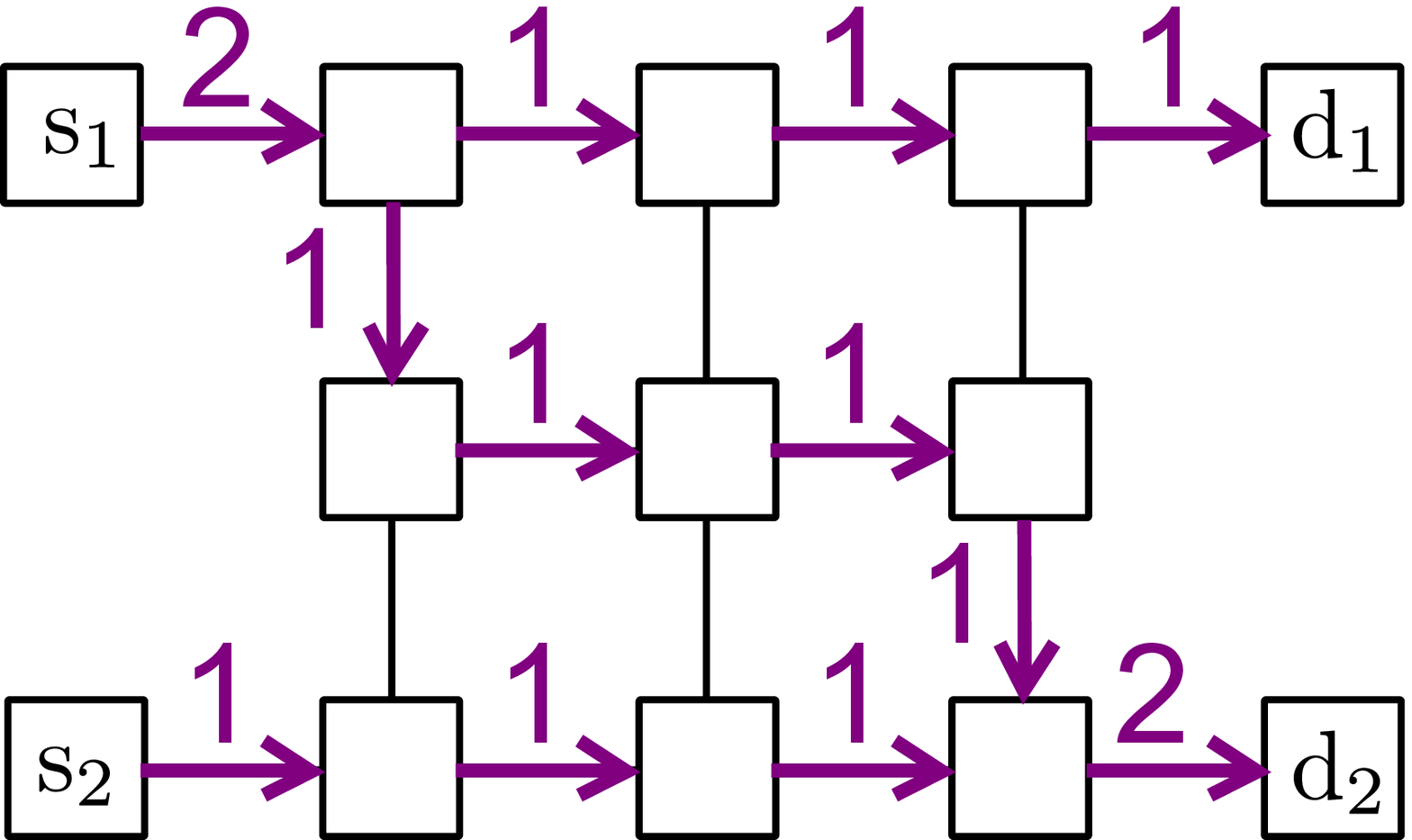}\label{fig:setting1-4}}
 \hfil
 \subfloat[An optimal flow $\alpha_5$]{\includegraphics[width=0.25\linewidth]{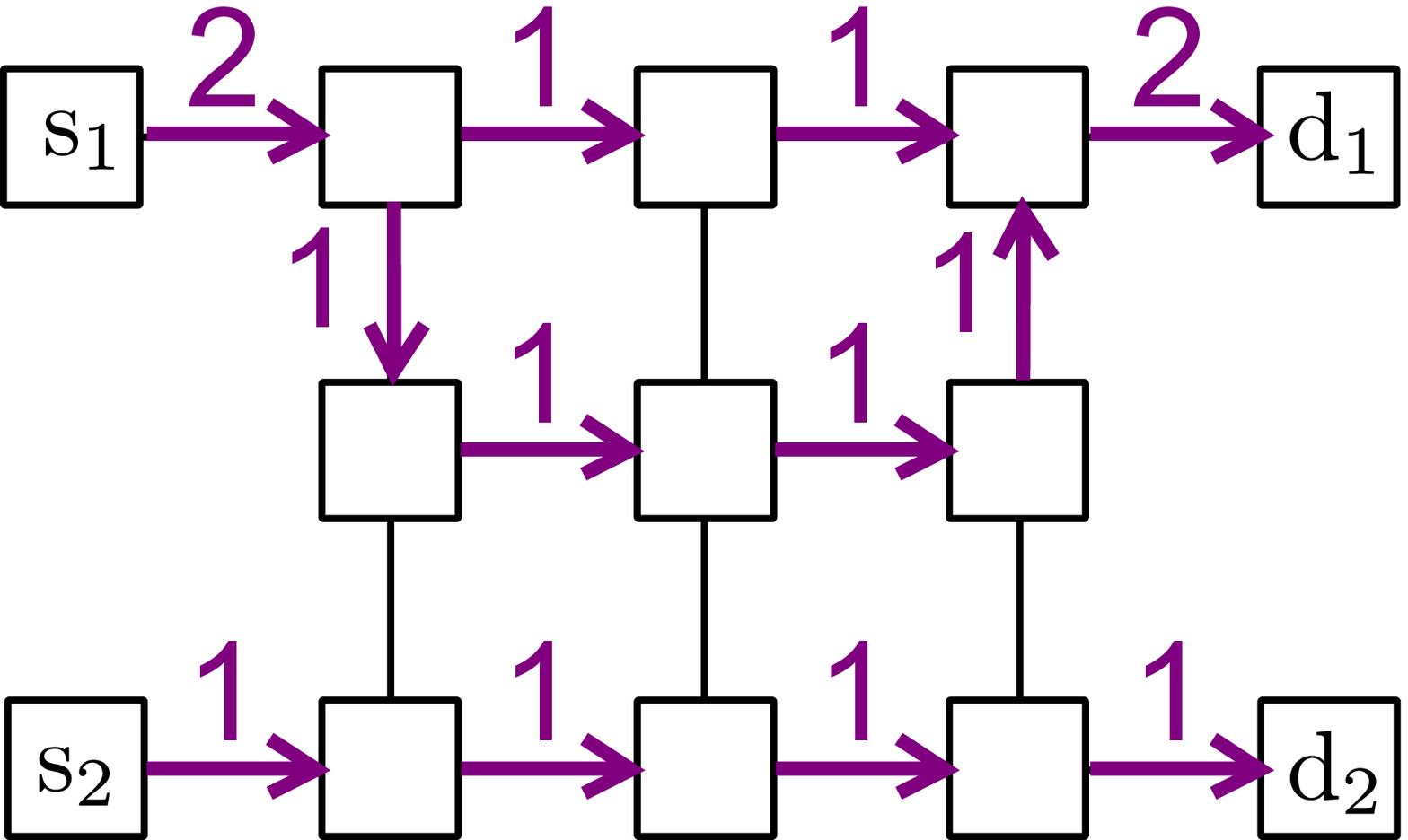}\label{fig:setting1-5}} 

 \caption{Optimal flows with $\Gamma_1$ and $N=1$. The energies are given as (a) $U_1({\rm s}_1)=0$, $U_1({\rm s}_2)=0$, $U_1({\rm d}_1)=0$, and $U_1({\rm d}_2)=0$, (b) $U_1({\rm s}_1)=1$, $U_1({\rm s}_2)=0$, $U_1({\rm d}_1)=-1$, and $U_1({\rm d}_2)=0$, (c) $U_1({\rm s}_1)=1$, $U_1({\rm s}_2)=1$, $U_1({\rm d}_1)=-1$, and $U_1({\rm d}_2)=-1$, (d) $U_1({\rm s}_1)=1$, $U_1({\rm s}_2)=1$, $U_1({\rm d}_1)=0$, and $U_1({\rm d}_2)=-2$, (e) $U_1({\rm s}_1)=2$, $U_1({\rm s}_2)=1$, $U_1({\rm d}_1)=-1$, and $U_1({\rm d}_2)=-2$, and (f) $U_1({\rm s}_1)=2$, $U_1({\rm s}_2)=1$, $U_1({\rm d}_1)=-2$, and $U_1({\rm d}_2)=-1$.} 
 \label{fig:setting1}
\end{figure*}
\begin{table*}[!ht]
 \centering
  \caption{Optimal flows $u$ in E1 and their costs with various settings of given energy $U_1$.}
  \label{tab:setting1}
  \vspace{1mm}
  \begin{tabular}[tb]{ccccc}
   \hline \hline
   $N$ & $U_1({\rm s}_1)$, $U_1({\rm s}_2)$ &  $U_1({\rm d}_1)$, $U_1({\rm d}_2)$ & Cost $\Gamma_1(u)$ & An Optimal Flow\\
   \hline
   $1$ &  $0$, $0$ & $0$, $0$ & $0$   & $\alpha_{0}$\\
   $1$ &  $1$, $0$ & $-1$, $0$ & $2$   & $\alpha_{1}$\\
   $1$ &  $1$, $1$ & $-1$, $-1$ & $4$   & $\alpha_{2}$\\
   $1$ &  $1$, $1$ &  $0$, $-2$  & $6$   & $\alpha_{3}$\\
   $1$ &  $2$, $1$ &  $-1$, $-2$  & $8$   & $\alpha_{4}$\\
   $1$ &  $2$, $1$ &  $-2$, $-1$  & $8$   & $\alpha_{5}$\\
   \hline \hline
  \end{tabular}
\end{table*}

 \begin{table}[t]
 \centering
  \caption{Optimal flows $u$ in E2 and their costs with various settings of $N$ and given energy $U_2$. Here, optimal flows are denoted by the sequence of  $\{\alpha_{i}\}_{i=0}^{5}$ in time order.}
  \label{tab:setting2}
  \begin{tabular}[tb]{cccc}
   \hline \hline
   $N$ &  $U_2({\rm d}_1)$, $U_2({\rm d}_2)$ & Cost $\Gamma_2(u)$ & An Optimal Flow\\
   \hline
   $2$ &  $3$, $3$ & $16$ & $\alpha_{4}\alpha_{5}$\\
   $3$ &  $3$, $3$ & $12$ & $\alpha_{2}\alpha_{2}\alpha_{2}$\\
   $5$ &  $3$, $3$ & $12$ & $\alpha_{2}\alpha_{2}\alpha_{2}\alpha_{0}\alpha_{0}$\\
   $2$ &  $5$, $2$ & $1016$ & $\alpha_{5}\alpha_{5}$\\
   $3$ &  $5$, $2$ & $18$ & $\alpha_{5}\alpha_{5}\alpha_{1}$\\
   $5$ &  $5$, $2$ & $14$ & $\alpha_{2}\alpha_{2}\alpha_{1}\alpha_{1}\alpha_{1}$\\
   \hline \hline
  \end{tabular}
\end{table}

 \begin{table}[t]
 \centering
  \caption{Optimal flows $u$ in E3 and their costs with various settings of $N$ and given energy $U$. Here, optimal flows are denoted by the sequence of  $\{\alpha_{i}\}_{i=0}^{5}$ in time order. The settings of $N$ and $U$ correspond to the settings in Tab.~\tabref{tab:setting2} through the relationship $\sum_{d\in V_{\rm D}} U_2(d) = U$.}
  \label{tab:setting3}
  \begin{tabular}[tb]{cccc}
   \hline \hline
   $N$ &  $U$ & Cost $\Gamma_3(u)$ & An Optimal Flow\\
   \hline
   $2$ &  $6$ & $16$ & $\alpha_{4}\alpha_{4}$\\
   $3$ &  $6$ & $12$ & $\alpha_{2}\alpha_{2}\alpha_{2}$\\
   $5$ &  $6$ & $12$ & $\alpha_{2}\alpha_{2}\alpha_{2}\alpha_{0}\alpha_{0}$\\
   $2$ &  $7$ & $1016$ & $\alpha_{4}\alpha_{4}$\\
   $3$ &  $7$ & $16$ & $\alpha_{4}\alpha_{2}\alpha_{2}$\\
   $5$ &  $7$ & $14$ & $\alpha_{2}\alpha_{2}\alpha_{2}\alpha_{1}\alpha_{0}$\\
   \hline \hline
  \end{tabular}
\end{table}

Now, we discuss these examples, solving them using cycle-canceling algorithm \cite{bib:murota2003b}. 
First, with $\Gamma_1$ in E1, we set $N = 1$ and show an optimal flow for each setting of $U_1$ in Fig.~\ref{fig:setting1}.
The six optimal flows are named as $\{\alpha_{i}\}_{i=0}^{5}$.
The optimal flows $u$ are listed in Tab.~\tabref{tab:setting1} with their settings and costs $\Gamma_1(u)$. 
Next, with $\Gamma_2$ in E2 and with $\Gamma_3$ in E3,
we list optimal flows as shown in Tabs.~\tabref{tab:setting2} and \tabref{tab:setting3}, respectively.
Here, each optimal flow is denoted by a sequence of $\{\alpha_{i}\}_{i=0}^{5}$ in time order.
Note that the settings of $N$ and $U$ in Tab.~\tabref{tab:setting3} correspond to the settings in Tab.~\tabref{tab:setting2} through the relationship $\sum_{d\in V_{\rm D}} U_2(d) = U$.
In the optimal flows in which costs exceed $1000$, the given energy is not represented.

From Fig.~\figref{fig:setting1} and Tab.~\tabref{tab:setting1}, the following properties are confirmed at $N=1$:
\begin{itemize}
 \item  By minimizing the cost of V1, given energy $U_1$ is represented at each source and each destination.
 \item By minimizing the cost of V2, energy is transferred with the smallest number of links. Note that, in flow $\alpha_2$ in Fig.~\figref{fig:setting1-2}, the destinations ${\rm d}_1$ and ${\rm d}_2$ are supplied from the nearest sources, i.e. ${\rm s}_1$ for ${\rm d}_1$ and ${\rm s}_2$ for ${\rm d}_2$.
 \item By imposing the constrain on V3 as shown in \eqref{eq:example_dom_f}, energy is transferred without strain, i.e. without change of stored energy.
\end{itemize}
Tables~\tabref{tab:setting2} and \tabref{tab:setting3} imply the following properties:

\begin{itemize}
 \item The cost of the optimal flow decreases as the number of time step, $N$, increases up to a certain point. 
       From $N$ exceeding the point, the cost takes a constant value and $\alpha_0$, in which symbols are not transferred, is added to the optimal flow.  

 \item The cost of the optimal flow decreases when the distribution of supplied energy is not specified and the total supplied energy is given. 
       For example,  the cost becomes $\Gamma_2(u) = 18$ when we set $N=3$, $U_2({\rm d_1})=5$, and $U_2({\rm d_2})=2$ in E2, while the cost becomes $\Gamma_3(u)=16$ when we set $N=3$ and $U=7$, which is equal to $\sum_{d\in V_{\rm D}} U_2(d)=5+2$, in E3.
      
\end{itemize}
In other words, a more transferable SPM can be selected if symbols are transferred with less temporal and spatial restriction.
These properties show that power can be packetized and be controllable while preserving reasonable properties of power.

\section{Conclusion}

In this paper, we have established the packet-centric framework of electrical energy networks, defining symbols in power packetization as a minimum unit of electric power transferred by a power pulse with an information tag.
Here, packetized power is spatially and temporally transferred as symbols in a digitized and quantized manner.
At each source and each destination, given energy is represented during a finite duration with symbols as the total amount of energy of symbols.
Then, the supply and the demand are met in the network.

To mathematically represent such transmission of packetized power, 
we introduced the symbol propagation matrix, in which a symbol is transferred at a link during a unit time.
Via SPM, packetized power is described as a network flow in a spatio-temporal structure.
Then, we considered a network flow problem for selecting an SPM in terms of transferability, that is, the possibility to represent given energies at sources and destinations during the finite duration.
In networks, packetized power appears as supplied energy from sources and supplied energy to destinations (V1), transferred energy at each link during each unit time (V2), and change of stored energy in each router (V3).
Setting a laminar family of subsets of nodes in spatio-temporal structure for the costs of V1 and V3, we can formulate this problem as an M-convex submodular flow problem which is known as generalization of the minimum cost flow problem and solvable.
Unlike conventional minimum cost flow problems, here, we weighted not only values of network flow (V2) but also values of boundary of network flow and their time integrals (V1 and V3). 
Finally, the formulation was discussed through examples and it is shown that power can be packetized and be controllable while preserving reasonable properties of power. 

The established packet-centric framework is completely different from the circuit theory, in which power is handled in a continuous manner and is governed by Kirchhoff laws and Tellegen's theorem \cite{bib:desoer1969b}. 
Here, power packet is introduced as a unit of electric power, so that power is digitized and quantized. 
The results of this paper suggest a mathematical framework which integrates energy and information in electrical energy networks.

\section*{Acknowledgments}

The author (S.N.) would like to thank the current and former members of the Robotics, Perception and Learning Laboratory of the Royal Institute of Technology (KTH) for fruitful discussions.

Parts of this work were financially supported by the Cross-Ministerial Strategic Innovation Program from the New Energy and Industrial Technology Development Organization, Japan, and by the Super Cluster Program (Kyoto) from the Japan Science and Technology Agency. The work of the author (S.N.) was financially supported in part by Kyoto University.











\end{document}